\documentclass[10pt,article,letterpaper]{IEEEtran}
\usepackage[utf8]{inputenc}
\usepackage{subcaption}
\date{November 2020}
\usepackage{amsmath, amssymb,amsthm,cite}
\usepackage{graphicx}
\usepackage{psfrag}
\usepackage{bbm}
\usepackage{mathtools}
\usepackage{color}
\usepackage{tcolorbox}
\usepackage{algorithm,algpseudocode}
\usepackage{bm}
\usepackage{geometry}
\DeclareMathOperator*{\nn}{\nonumber}

\renewcommand{\c}[1]{(zI-#1)^{-1}}
\renewcommand{\a}[1]{(z^{-1}I-#1)^{-1}}

\allowdisplaybreaks
\newtheorem{lemma}{Lemma}
\newtheorem{problem}{Problem}
\newtheorem{theorem}{Theorem}

\theoremstyle{definition}

\newtheorem{assumption}{Assumption}

\makeatletter
\def\blfootnote{\gdef\@thefnmark{}\@footnotetext}
\makeatother

\title{Regret-Optimal Control under Partial Observability}
\author{Joudi Hajar\thanks{J. Hajar and B. Hassibi are with the Department of Electrical Engineering at Caltech (emails: \{jhajar,hassibi\}@caltech.edu).}, Oron Sabag\thanks{O. Sabag is with the School of Computer Science and Engineering at the Hebrew University of Jerusalem (email: oron.sabag@mail.huji.ac.il).}\thanks{The work of O. Sabag was conducted at Caltech during 2020-2022. The work of J. Hajar was done at Caltech in 2023.}, Babak Hassibi}
\date{July 2021}

\begin{document}
\maketitle
\begin{abstract}
 This paper studies online solutions for regret-optimal control in partially observable systems over an infinite-horizon. Regret-optimal control aims to minimize the difference in LQR cost between causal and non-causal controllers while considering the worst-case regret across all $\ell_2$-norm-bounded disturbance and measurement sequences. Building on ideas from \cite{sabag2021regret} on the the full-information setting, our work extends the framework to the scenario of partial observability (measurement-feedback). We derive an explicit state-space solution when the non-causal solution is the one that minimizes the $\mathcal H_2$ criterion, and demonstrate its practical utility on several practical examples. These results underscore the framework's significant relevance and applicability in real-world systems. 
\end{abstract}
\section{Introduction}\label{sec:intro}
Regret-optimal control stands as a robust framework for designing controllers capable of dealing with uncertainty within dynamic systems through the lens of competitive analysis. The main idea is to design a controller that minimizes the disparity between its own performance and that of a reference controller that has a superior performance. Typical competitive metrics are the regret and the competitive ratio, defined as the difference or ratio between the controller's performance and the reference controller's performance. One of the fundamental questions studied in the current manuscript is how the designed policy will be affected by the choice of the reference controller, an impractical entity that is aimed to be mimicked.

The regret optimal control framework was introduced for the full-information control setup where the designed controller has a strictly causal access to the system's state and the reference controller is the one that has non-causal access to the states \cite{sabag2021regret} (and its extended version \cite{sabag2023regretoptimal}). Interestingly, the optimal regret and, as a result, the regret optimal controller can be derived explicitly in the infinite horizon scenario via a reduction to the Nehari problem \cite{nehari1957bounded}. This is in contrast to typical solutions in robust control that involve an additional optimization over the corresponding sub-optimal problem. The study of this framework was then extended to the finite horizon setup with and without safety constraints \cite{goel2020regret,pmlr-v168-martin22a,didier2022system}, and to the competitive ratio metric \cite{sabag2022optimal,goel2022competitive}. In all of those works (see also \cite{karapetyan2022regret,martin2022follow,liu2023robust,hajar2023wasserstein,yan2023distributionally,sabag2022regret}), regret has been consistently shown to be a compelling metric for controllers' design as it nicely interpolates the robust and average performance of $\mathcal H_\infty$ and $\mathcal H_2$ control. Motivated by this, we consider the regret optimal control approach to partially observable control systems.

For the full information control setup, a natural candidate for the reference controller is a clairvoyant controller that has a non-causal access to the entire disturbances sequence (equivalently, the state generated by the uncontrolled system) before hand. This clairvoyant controller is unique and outperforms any controller for any disturbance sequence \cite{blackbook}. 
For partially observable systems, the clairvoyant controller has a non-causal access to the system's measurements generated by the uncontrolled dynamical system. However, there does not exist a unique controller that outperforms any other controller \emph{for any noise sequence}. Therefore, the clairvoyant controller has to be optimal under particular metric e.g. $\mathcal H_2$ or $\mathcal H_\infty$. Our first result shows that the regret-optimal control problem can be reduced to a Nehari problem no matter what is the reference controller. This is a generalization of \cite{goel21L4DC} for the $\mathcal H_2$ clairvoyant controller and it elucidates the structure of the regret-optimal controller. Previously, it was shown that the regret optimal control problem can be reduced to an $\mathcal H_\infty$ problem if the reference controller has a direct access to the system disturbance \cite{goel2022measurement}. In contrast, our reference controller only assumes non-causal access to the measurements. This implies that the operational comparison between the LQR costs is made over both the disturbance and the measurement sequences. The reduction to the Nehari problem also serves as the first step towards an explicit solution for the state-space setting as the Nehari solution has an explicit solution \cite{sabag2023regretoptimal}.

%We demonstrate that this is a central element when designing a competitive-analysis based system. If the reference controller is too optimistic, the comparison effect via regret vanishes. 

For the state space setting, we focus on the non-causal $\mathcal H_2$ controller as our reference for the regret metric. Through spectral factorizations and decompositions, we are able to derive an explicit regret-optimal controller. The controller is represented by a finite-dimensional and time invariant state-space whose inputs are the system measurements. Thus, it can be implemented in real-time. This opens up new avenues for practical application of regret-optimal control, with potential impacts in various fields, including robotics, autonomous vehicles, and control systems design.

We summarize our main contribution:
\begin{enumerate}
    \item We study regret-optimal control with partial observability and present a reduction to the Nehari problem for \emph{any reference controller}. This is presented in Theorem~\ref{th:reduction}, and holds for any time horizon.
    \item For the state-space setting, we derive an explicit solution for regret-optimal control under partial observability. This solution applies to the doubly infinite horizon regime and the reference controller is the $\mathcal H_2$ non-causal controller that has access to the system measurements. To the best of our knowledge, this is the first solution of partial observability in the context of regret optimal control that involves a comparison of control costs (of the designed and the reference controllers) that depend on \emph{two sequences}, the disturbances and the measurements noise. 
\end{enumerate}
% The study of the finite-horizon regime enables one to model time-varying dynamical systems but on the other hand it is feasible to solve the optimization problem for relatively short horizons. 

To validate the effectiveness of our approach, we conduct simulations on a range of systems, demonstrating the efficacy of our proposed method in various settings. We organize the paper as follows: Section \ref{sec:setting} provides the preliminaries, Section \ref{sec::main1} presents the reduction of the regret control problem to a Nehari problem, Section \ref{sec::main2} offers the state space solution of the controller, and Section \ref{sec::sim} presents the simulations. All proofs are to be found in the appendix of the extended version of this paper available online.
%(Appendix \ref{sec:proofs}). 
%While in this paper, our primary focus lies on the scenario where the regret-optimal controller we design is causal, we provide the state-space solution for the case of a strictly causal controller in Appendix \ref{sec::SCsol}.

\section{Problem Definition and Preliminaries}\label{sec:setting}

\subsection{Notation}\label{subsec:}
$\mathbb{R}$ denotes the set of real numbers, $\| \cdot\|_2$ is the 2-norm, $\| \cdot\|$ is the operator norm, $^\ast$ stands for adjoint, $z$ is the complex variable of the $z$-domain, and $I$ is the identity matrix. A calligraphic letter represents a doubly infinite operator, and its corresponding standard letter denotes the transfer matrix of the operator.

\subsection{The control setting}

We examine a discrete-time, linear time-invariant (LTI) dynamical system described by the following state-space model:
\begin{align}\label{eq::ss} 
    x_{i+1}&= Fx_i + G_1w_i + G_2 u_i\nn\\
    s_i &= Lx_i \nn\\
    y_i &= Hx_i + v_i.
\end{align}
In this representation, $x_i\in \mathbb{R}^n$ denotes the system's state, $u_i\in \mathbb{R}^m$ represents the control input, $w_i \in \mathbb{R}^n$ is the process noise, $s_i \in \mathbb{R}^q$ is the regulated signal while $y_i \in \mathbb{R}^p$  represents the noisy state measurements, and $v_i \in \mathbb{R}^p$ is the measurement noise.  

Throughout this paper, we adopt an operator form representation of the system dynamics~\eqref{eq::ss}. Considering a  \emph{doubly infinite horizon}, let us define\\ 
\( ~~~~~~~~~~~~s \coloneqq \left( \begin{array}{c} \vdots \\ s_{-1} \\ s_0 \\ s_1\\ \vdots\end{array} \right) ~~~,~~~
u \coloneqq \left( \begin{array}{c} \vdots \\ u_{-1} \\ u_0 \\u_1\\ \vdots \end{array} \right),~~~~~~~~~~~~\) and similarly for $w$, $x$, $y$, and $v$. With these definitions, we can express the system dynamics~\eqref{eq::ss} equivalently using an operator-based formulation:
\begin{align}\label{eq::of}
    s &= \mathcal P_{11}w + \mathcal P_{12}u \nn\\
    y &= \mathcal P_{21}w + \mathcal P_{22}u + v\nn\\
    u &= \mathcal K y,
\end{align}
where the operators $\mathcal P_{ij}$ for $i,j\in\{1,2\}$ are strictly causal time-invariant operators (i.e, strictly lower triangular block Toeplitz operators) corresponding to the dynamics~\eqref{eq::ss}, and $\mathcal K$ is a linear causal controller to be designed. 

We consider the Linear-Quadratic Gaussian (LQG) cost given as
\begin{equation}\label{eq::LQRcost1}
    J(u,w,v)\coloneqq s^\ast \bar{Q}s+u^\ast \bar{R}u
\end{equation}
where $\bar{Q}, \bar{R}$ are block diagonal operators with $Q,R\succ 0$ on their diagonals, respectively. To simplify the notation, we redefine $s$ and $u$ as $s\leftarrow \bar{Q}^{\frac{1}{2}}s$, and $u\leftarrow \bar{R}^{\frac{1}{2}}u$, which results in~\eqref{eq::LQRcost1} taking the form:
\begin{align} \label{eq::LQRcost2}
    J(u,w,v) &=s^\ast s + u^\ast u = \begin{pmatrix}
    w^\ast & v^\ast
    \end{pmatrix}T_\mathcal K^\ast T_\mathcal K \begin{pmatrix}
    w \\ v  
    \end{pmatrix}.
\end{align}
In this expression,  $T_\mathcal{K}$ represents the transfer operator associated with $\mathcal K$, mapping the input sequences $\begin{pmatrix}
    w\\v
\end{pmatrix}$ to the regulated state and control sequences $\begin{pmatrix}
    s\\u
\end{pmatrix}$, i.e., $$    T_{\mathcal K}: \begin{pmatrix}
    w\\v
    \end{pmatrix}\to \begin{pmatrix}
    s\\u
    \end{pmatrix}.$$
% \begin{align}
%     T_{\mathcal K}: \begin{pmatrix}
%     w\\v
%     \end{pmatrix}\to \begin{pmatrix}
%     s\\u
%     \end{pmatrix}.
% \end{align}
The operator $T_\mathcal{K}$ can be computed explicitly as
{\small
\begin{align}
    T_{\mathcal K}&= \begin{pmatrix}
    \mathcal P_{11} + \mathcal P_{12}\mathcal K(I - \mathcal P_{22}\mathcal K)^{-1} \mathcal P_{21}&\mathcal P_{12}\mathcal K(I - \mathcal P_{22}\mathcal K)^{-1}  \\
\mathcal K(I - \mathcal P_{22}\mathcal K)^{-1} \mathcal P_{21}& \mathcal K(I - \mathcal P_{22}\mathcal K)^{-1}
    \end{pmatrix}.\nonumber
\end{align}}
As $\mathcal P_{22}$ and $\mathcal P_{22}\mathcal K$ are strictly causal, we can define the Youla parameterization:
\begin{equation}
    \mathcal Q = \mathcal K (I - \mathcal P_{22}\mathcal K)^{-1}
\end{equation} 
or equivalently, 
\begin{equation}
    \mathcal K = \mathcal Q(\mathcal I + \mathcal Q \mathcal P_{22})^{-1}
\end{equation}
which implies that $\mathcal K$ is causal (/strictly-causal) if and only if $\mathcal Q$ is causal (/strictly-causal) \cite{blackbook}. $T_\mathcal K$ can be rewritten in terms of $\mathcal Q$ as
\begin{align}
        T_{\mathcal Q}&= \begin{pmatrix}
    \mathcal P_{11} + \mathcal P_{12}\mathcal Q \mathcal P_{21}&\mathcal P_{12}\mathcal Q  \\
\mathcal Q \mathcal P_{21}& \mathcal Q
    \end{pmatrix}.
\end{align}

\subsection{Regret-optimal control under partial observability}
Given a benchmark non-causal controller $\mathcal{Q}_0$, we define the regret as:
\begin{align}
    R(\mathcal{Q},w,v)&\coloneqq J(\mathcal{Q},w,v)- J(\mathcal{Q}_0,w,v) \nn\\
    &= \begin{pmatrix}
    w^\ast & v^\ast
\end{pmatrix} (T_\mathcal{Q}^\ast T_\mathcal{Q}-T_{\mathcal{Q}_0}^\ast T_{\mathcal{Q}_0})\begin{pmatrix}
    w \\
    v\\
\end{pmatrix}\label{eq::regret},
\end{align}
which quantifies the additional cost incurred by a causal controller due to its lack of knowledge about future disturbances. In simpler terms, regret represents the disparity between the cost accrued by a causal controller and the cost accumulated by a reference non-causal controller equipped with complete foresight of the disturbance trajectory. For a reference controller $\mathcal Q_0$, the regret-optimal control problem is  defined below:
\begin{problem}[The Regret Problem]\label{pb::RO}
The problem of minimizing regret in the measurement-feedback, infinite-horizon setting is:
\begin{equation}\label{eq::ROMF}
\inf_{\mathcal{Q}} \sup_{\substack{ w,v}}\frac{R(\mathcal{Q},w,v)}{\|w\|_2^2+ \|v\|_2^2},
\end{equation} with $R(\mathcal{Q},w,v)$ as in \eqref{eq::regret}.
\end{problem}
The goal in the next sections is to solve explicitely the regret problem, first by reducing it to a Nehari problem (Section \ref{sec::main1}) and then finding a state-space solution for the controller $\mathcal{Q}$ (Section \ref{sec::main2}).

\section{Main result I: Reduction to a Nehari Problem} \label{sec::main1}
In this section, we focus on reducing the regret problem to a Nehari problem, facilitating the derivation of an explicit solution for the causal regret-optimal controller $\mathcal Q_c$. We will later transform this solution into a state-space representation in Section \ref{sec::main2}. 

The Nehari problem aims to find the optimal causal approximation for an operator that is strictly anti-causal, as measured by the operator norm:
\begin{problem}[The Nehari Problem] Find a causal block Toeplitz operator $\mathcal{Y}$ that minimizes $\|\mathcal Y- \mathcal Z\|$, where $\mathcal{Z}$ is a given strictly anti-causal block Toeplitz operator.
\end{problem}
Theorem 5 of \cite{sabag2023regretoptimal} demonstrates that in cases where the operator has a state-space structure, we can explicitly determine both the minimal norm and the approximation $\mathcal Y$, as we will endeavor in Section \ref{sec::main2}.

Now, to tackle the regret problem, we first need to choose the reference (non-causal) controller, $Q_0$. $Q_0$ should be chosen as an operationally-superior controller with respect to the one to be designed. Otherwise, the designed controller can be chosen as $Q_0$ and the optimal regret is zero. We first present a general result which shows that the regret problem with respect to \emph{any reference controller} can be reduced to a Nehari problem. 
\begin{theorem}[Reduction to a Nehari problem]\label{th:reduction}
A $\gamma$-solution to the causal regret problem with respect to $\mathcal Q_0$, i.e, a suboptimal solution to Problem \ref{pb::RO} that satisfies
\begin{equation}\label{eq::gammaSol}
\sup_{\substack{ w,v}}\frac{R(\mathcal{Q},w,v)}{\|w\|_2^2+ \|v\|_2^2}\leq\gamma^2,
\end{equation}
exists if and only if there exists a solution to the Nehari problem
\begin{align}\label{eq:th_reduction_Nehari}
    \inf_{\text{causal} \ \mathcal C} \| \mathcal C - \{\mathcal X\}_- \|\le 1
\end{align}
where $\{\mathcal X\}_-$ is defined to be the strictly anti-causal part of the operator 
\begin{equation}\label{eq::opX}
    \mathcal X = \mathcal M^{1/2}\mathcal Q_0\mathcal W^{-/2} + \gamma^{-2}\mathcal M^{-\ast/2} \mathcal T (\mathcal  Q_2 - \mathcal Q_0)\mathcal W^{-/2},
\end{equation}  and can be computed from the canonical factorizations
\begin{align}
     \mathcal M&=\mathcal M^{\ast/2}\mathcal M^{1/2} \nn \\
     &= \gamma^{-2}\mathcal I + \gamma^{-2} \mathcal P_{12}^\ast(\mathcal I +  \gamma^{-2}\mathcal P_{11}
     (\mathcal I + \mathcal P_{21}^\ast \mathcal P_{21})^{-1} \mathcal P_{11}^\ast)  \times \nn \\ & \ \ \ \ \mathcal P_{12},\label{eq::opM}\\
     \mathcal W &=\mathcal W^{\ast/2}\mathcal W^{1/2} \nn\\
     &=  (I + \mathcal P_{21} \mathcal P_{21}^\ast)^{-1} + \gamma^{-4} (\mathcal Q_2 - \mathcal Q_0)^\ast \mathcal T \mathcal M^{-1}\mathcal T \times \nn \\ &\ \ \ \ (\mathcal Q_2 - \mathcal Q_0) \label{eq::opW}.
\end{align}
Here, 
\begin{equation}
    \mathcal T= I+\mathcal P_{12}^\ast\mathcal P_{12}
\end{equation}
and 
\begin{equation}\label{eq::Q2}
    \mathcal Q_2 = - {\mathcal T}^{-1}\mathcal P_{12}^\ast \mathcal P_{11}\mathcal P_{21}^\ast(I + \mathcal P_{21}\mathcal P_{21}^\ast)^{-1}
\end{equation}
is the optimal non-causal $\mathcal H_2$ controller that minimizes the Frobenius norm of $T_\mathcal Q$ \cite{blackbook}.
Moreover, the optimal regret is equal to $\gamma^2$ where $\gamma$  achieves equality in \eqref{eq::gammaSol}. If a solution $\mathcal C$ to \eqref{eq:th_reduction_Nehari} exists, then the regret-optimal controller is given by 
\begin{align}
    \mathcal Q_c = \mathcal M^{-/2} (\mathcal C  + \{\mathcal X\}_+)\mathcal W^{1/2}.\label{eq::optCont}
\end{align}
\end{theorem}
The proofs for Theorem \ref{th:reduction} and all subsequent lemmas and theorems can be found in Appendix B of the extended paper online. We also remark that the reduction is not limited to the doubly-infinite horizon regime and holds for the finite and one-sided infinite regimes as well. The special case of $\mathcal Q_0 = \mathcal Q_2$ in Theorem \ref{th:reduction} appeared in \cite{goel21L4DC}. Theorem \ref{th:reduction} shows that the control problem is equivalent to a Nehari problem, and thus the minimum norm can be determined if one is able to perform the required factorizations. 

However, formulating an implementable controller is not straightforward since canonical factorizations cannot be directly computed for high-dimensional operators. In such cases, the reference controller can be selected based on its superiority in the time-domain, controllability space, or observability. Remarkably, the regret problem can be reduced to a Nehari problem in all these scenarios, even though the Nehari problem primarily deals with operator comparisons in the time-domain. A specific instance of Theorem \ref{th:reduction} occurs when the reference controller aligns with the optimal non-causal $\mathcal{H}_2$ controller.

We proceed to provide an explicit solution for the regret-optimal controller in the infinite-horizon state-space control problem while focusing on the the special case $\mathcal Q_0 = \mathcal Q_2$. In this case, the operators $\mathcal X$ \eqref{eq::opX} and $\mathcal W$ \eqref{eq::opW} simplify to 
\begin{align}
 \mathcal X &= \mathcal M^{1/2}\mathcal Q_2\mathcal W^{-/2} \label{eq::X-z}\\
\mathcal W&=\mathcal W^{\ast/2}\mathcal W^{1/2} = \mathcal (I + \mathcal P_{21} \mathcal P_{21}^\ast)^{-1} \label{eq::W-z}.
\end{align}

\section{Main result II: The state-space model}\label{sec::main2}
In this section, our primary objective is to establish the state-space formulation of the controller. To achieve this, we begin in Section \ref{subsec:zdomain} by representing the operators in the $z$-domain. We then proceed to Section \ref{subsec:Factorization}, where the main focus is on factorizing $\mathcal W$ \eqref{eq::opW} and $\mathcal{M}$ \eqref{eq::opM}. Subsequently, in Section \ref{subsec:decomposition}, we delve into the decomposition of $\mathcal X$ \eqref{eq::opX} into causal and strictly-anti-causal operators. Finally, in Section \ref{subsec:finalsol}, we leverage these factorizations and operator decompositions to provide both the explicit solution to the Nehari problem and a state-space representation of the regret-optimal controller.
\subsection{Operators in $z$-domain}\label{subsec:zdomain}
The operators in the general problem \eqref{eq::of} take the form of Toeplitz operators that can be described concisely using the transfer matrices in the $z$-domain:
\begin{alignat}{2}
    P_{11}(z) &\triangleq L\c{F}G_1, &\quad P_{12}(z) &\triangleq L\c{F}G_2 \nn\\
    P_{21}(z) &\triangleq H\c{F}G_1, &\quad P_{22}(z) &\triangleq H\c{F}G_2.
\end{alignat}
Also, the solution can be expressed using the $z$-domain counterpart of the operators \eqref{eq::opX}-\eqref{eq::opW}, and \eqref{eq::Q2} defined in Theorem \ref{th:reduction}, as follows:
\begin{align}
   X(z) &= M^{1/2}(z) Q_2(z) W^{-/2}(z) \label{eq::tfX}\\
   M(z) &= M^{\ast/2}(z^\ast)M^{1/2}(z)=\gamma^{-2} I + \gamma^{-2} P_{12}^\ast(z^{-\ast}) (I + \nn\\
 & \ \ \ \ \gamma^{-2}  P_{11}(z)  ( I + P_{21}^\ast (z^{-\ast}) P_{21}(z))^{-1} P_{11}^\ast(z^{-\ast}) )\times\nn\\ &\ \ \ \ P_{12}(z) \\
W(z)&=W^{\ast/2}(z^{\ast})W^{1/2}(z)=\mathcal (I + P_{21}(z) P_{21}^\ast(z^{-\ast}))^{-1}\\
 Q_2(z)&= - ( I + P_{12}^\ast(z^{-\ast}) P_{12})^{-1}(z)  P_{12}^\ast(z^{-\ast}) P_{11}(z)\times\nn \\ &\ \ \ \  P_{21}^\ast(z^{-\ast}) (I + P_{21}(z) P_{21}^\ast(z^{-\ast}))^{-1}
\end{align}
We also define the operator 
\begin{equation}\label{eq::opS}
    \mathcal S = \mathcal I + \mathcal P_{12} \mathcal P^\ast_{12}
\end{equation}
 with its transfer matrix
\begin{align}\label{eq::Stf}
    S(z)&\triangleq I + L\c{F}G_2G_2^\ast\c{F^\ast}L^\ast,
\end{align}
whose factorization will be used in the derivation of the solution to the Nehari problem. % in Lemma \ref{lemma:nehariSol}.
\subsection{Factorizations}\label{subsec:Factorization}
To be able to perform the required factorizations, we introduce the following assumption:
\begin{assumption}\label{ass:stab}
    Assume that ($F,H$) and ($F,L$) are detectable, ($F,G_1$) is stabilizable, and ($F,G_1$) and ($F,G_2$) are controllable on the unit circle.
\end{assumption}
For the rest of this paper, we rely on assumption \ref{ass:stab}, and its significance will become apparent when reviewing the proofs in the appendix. We now factorize the operator $\mathcal W$ \eqref{eq::opW} in Lemma \ref{lemma:factor_W}:
\begin{lemma}[Factorization of $\mathcal W$]\label{lemma:factor_W}
The transfer matrix \begin{equation}
    W^{-1}(z)=I + P_{21}(z) P_{21}^\ast(z^{-\ast})
\end{equation} can be factorized as
\begin{align}
    W^{-/2}(z)W^{-\ast/2}(z^{-\ast})&= I + H\c{F}G_1G_1^\ast \times \nn \\ &\ \ \ \ (z^{-1}I-F^\ast)^{-1}H^\ast
\end{align}
with 
\begin{align}
    W^{-/2}(z)&=(I + H \c{F}K_W) R_W^{1/2},
\end{align}
where 
\begin{align}
    K_W&= F P_U H^\ast R_W^{-1},\nn \quad  F_W= F-K_WH \\
    R_W&= I + H P_U H^\ast = R_W^{1/2}R_W^{\ast/2}, 
\end{align}
and $P_U$ is the stabilizing solution to the Riccati equation
\begin{align}
    P_U&= F P_U F^\ast + G_1G_1^\ast - K_W R_W K_W^\ast.
\end{align}
Finally,
\begin{align}
    W^{1/2}(z)&= R_W^{-/2}(I + H \c{F}K_W)^{-1} \nn \\
    &= R_W^{-/2}(I - H \c{F_W}K_W) \label{eq::Whalf} \\
    W^{\ast/2}(z^{-\ast})&= (I +K_W^\ast \a{F^\ast}H^\ast)^{-1}R_W^{-\ast/2},
\end{align}
and $W^{1/2}(z)$ is causal and stable (bounded on the unit circle).
\end{lemma}
Now, we present a series of helpful lemmas (Lemma \ref{lemma:factor_M}-\ref{lemma:factor_nabla}) that we use to get to the factorization of the operator $\mathcal M$ \eqref{eq::opM} in Lemma \ref{lemma:factor_M2}.

\begin{lemma}[Factorization of $\mathcal M$]\label{lemma:factor_M}
The operator $\mathcal M$ \eqref{eq::opM} can be directly factorized as
\begin{align}\label{eq::tfM}
        \mathcal M^{\ast/2} \mathcal M^{1/2}&= \gamma^{-2}(\mathcal I + \mathcal P_{12}^\ast \nabla^{-\ast}\nabla^{-1} \mathcal  P_{12}),
\end{align}
where $\mathcal M^{1/2}$ is a causal operator and $\Gamma,\nabla$ are causal operators obtained from
\begin{align}
    \Gamma^\ast \Gamma &= \gamma^{2} (I +\mathcal P_{21}^\ast \mathcal P_{21}) + \mathcal P_{11}^\ast \mathcal P_{11} \label{eq::gamma}\\
        \nabla \nabla^\ast &= \mathcal I - \mathcal P_{11} \Gamma^{-1}\Gamma^{-\ast} \mathcal P_{11}^\ast.\label{eq::nabla}
\end{align}
\end{lemma}
We proceed to factorize $\Gamma^\ast \Gamma$ \eqref{eq::gamma} and $\nabla \nabla^\ast$\eqref{eq::nabla} in Lemmas \ref{lemma:factor_gamma},\ref{lemma:factor_nabla} respectively.
\begin{lemma}\label{lemma:factor_gamma}
The transfer matrix corresponding to \eqref{eq::gamma},\begin{equation}
    \Gamma^\ast (z^{-\ast}) \Gamma(z)=\gamma^{2} (I + P_{21}^\ast(z^{-\ast}) P_{21}(z)) + P^\ast_{11}(z^{-\ast}) P_{11}(z),
\end{equation}  can be factorized as $\Gamma^\ast(z^{-\ast}) \Gamma(z)$ with
\begin{align}
    \Gamma(z)&= R_\Gamma^{1/2} (I + K_\Gamma \c{F}G_1),
\end{align}
where
\begin{align}
    K_\Gamma&= R_\Gamma^{-1}G_1^\ast P_\Gamma F, \quad F_\Gamma= F - G_1 K_\Gamma \nn\\
    R_\Gamma&= \gamma^2 I + G_1^\ast P_\Gamma G_1 = R_\Gamma^{\ast/2} R_\Gamma^{1/2},
\end{align}
and $P_{\Gamma}$ is the stabilizing solution to the Riccati equation
\begin{align}
    P_\Gamma&= F^\ast P_\Gamma F   + (\gamma^2 H^\ast H + L^\ast L) - K_\Gamma^\ast R_\Gamma K_\Gamma.
\end{align}
\end{lemma}

\begin{lemma}\label{lemma:factor_nabla}
The transfer matrix corresponding to \eqref{eq::nabla}, \begin{equation}
    \nabla(z)\nabla^\ast(z^{-\ast})=I -  P_{11}(z) \Gamma^{-1}\Gamma^{-\ast}  P_{11}^\ast(z^{-\ast}),
\end{equation}can be factorized as
\begin{align}
 \nabla(z) \nabla^\ast(z^{-\ast}) &= I - G(z) \Gamma^{-1}(z)\Gamma^{-\ast}(z^{-\ast}) G^\ast(z^{-\ast}),
\end{align}
with
\begin{align}
    \nabla(z)&= (I + L \c{F_\Gamma}K_\nabla) R_\nabla^{1/2},
\end{align}
where 
\begin{align}
    K_\nabla&= F_\Gamma P_\nabla L^\ast R_\nabla^{-1},\nn \quad F_\nabla= F_\Gamma - K_\nabla L\\
    R_\nabla&= I + L P_\nabla L^\ast = R_\nabla^{1/2} R_\nabla^{\ast/2},    
\end{align}
and $P_\nabla$ is the stabilizing solution to the Riccati equation
\begin{align}
    P_\nabla&= F_\Gamma P_\nabla F_\Gamma^\ast - G_1R_\Gamma^{-1}G_1^\ast - K_\nabla R_\nabla K_\nabla^\ast.
\end{align}
\end{lemma}
Now, we get to the desired factorization of $\mathcal M$ \eqref{eq::opM}:
\begin{lemma}\label{lemma:factor_M2}
The transfer matrix corresponding to \eqref{eq::tfM},\begin{equation}
    M(z) = \gamma^{-2}(I + P_{12}^\ast(z^{-\ast}) \nabla^{-\ast}(z^{-\ast})\nabla^{-1}(z) P_{12}(z)),
\end{equation} can be factored as $M^{\ast/2}(z^{-\ast}) M^{1/2}(z)$ with
\begin{align}
    M^{1/2}(z)&= R_M^{1/2} (I + K_M \c{F_E}G_E),
\end{align}

where
\begin{align}
    K_M&= R_M^{-1}G_E^\ast P_M F_E\nn \quad F_M= F_E - G_E K_M\\
    R_M &= \gamma^{-2}I + G_E^\ast P_M G_E = R_M^{\ast/2} R_M^{1/2},    
\end{align}
and $P_M$ is a solution to the Riccati equation
\begin{align}
    P_M&= F_E^\ast P_M F_E   + H_E^\ast H_E - K_M^\ast R_M K_M.
\end{align}
The constants are defined as
\begin{align}
    H_E &\triangleq \begin{pmatrix}
    R_\nabla^{-/2}L &R_\nabla^{-/2}L 
    \end{pmatrix}\nn\\
    F_E &\triangleq \begin{pmatrix}
F&0\\
- K_\nabla L & F_\nabla
\end{pmatrix}, \quad
G_E \triangleq 
    \begin{pmatrix}
    \gamma^{-1}G_2\\
    0 
    \end{pmatrix}.
\end{align}
Also,
\begin{align}
    M^{-1/2}(z)&= (I + K_M \c{F_E}G_E)^{-1}R_M^{-/2}\label{eq::thetaInv}
\end{align}
\end{lemma}

Finally, we factorize the operator $\mathcal{S}$ \eqref{eq::opS} in the following lemma that we leverage later in Section \ref{subsec:decomposition} to find the Nehari solution.
\begin{lemma}[Factorization of $\mathcal S$]\label{lemma:factorization_S}
The transfer matrix $S(z)$ \eqref{eq::Stf} can be factorized as $S^{1/2}(z)S^{\ast/2}(z^{-\ast})$
with 
\begin{align}
    S^{1/2}(z)&=(I + L \c{F}K_S) R_S^{1/2},
\end{align}
where 
\begin{align}
    K_S&= F P_S L^\ast R_S^{-1}, \quad F_S= F-K_SL,\nn\\
    R_S&= I + L P_S L^\ast = R_S^{1/2}R_S^{\ast/2},
\end{align}
and $P_S$ is the stabilizing solution to the Riccati equation
\begin{align}
    P_S&= F P_S F^\ast + G_2G_2^\ast - K_S R_S K_S^\ast.
\end{align}
Finally,
\begin{align}
    S^{-/2}(z)&= R_S^{-/2}(I + L \c{F}K_S)^{-1}
    %S^{-\ast/2}(z^{-\ast})&= (I +K_S^\ast \a{F^\ast}L^\ast)^{-1}R_S^{-\ast/2}.
\end{align}
and $S^{-/2}(z) = (S^{1/2}(z))^{-1}$ is causal and stable (bounded on the unit circle).
\end{lemma}
\subsection{Decomposition}\label{subsec:decomposition}
Here in Lemma \ref{lemma:decomposition}, we decompose $X(z)$ \eqref{eq::tfX} into its causal and strictly anti-causal parts, to be used in solving the Nehari problem \eqref{eq:th_reduction_Nehari} as we will see in Lemma \ref{lemma:nehariSol}.
\begin{lemma}\label{lemma:decomposition}
The transfer matrix $X(z)$ \eqref{eq::tfX} can be expressed as a sum of two causal and a strictly anti-causal transfer matrices given by
 \begin{align}
     C_1(z)&= - M^{1/2}(z) G_2^\ast U_1F_S [\c{F_S}F_SU_2 + U_2]\times \nn \\ &\ \ \ \ H^\ast R_W^{-\ast/2} \label{eq::C1}\\
     C_2(z)&= - R_M^{1/2} K_M [\c{F_E}F_EU_3 + U_3]G_A\label{eq::C2}\\
 A(z)&=  H_A\a{F_A^\ast}G_A \label{eq::A},
 \end{align}
 where $U_1,U_2,U_3$ are given as the solution to the Lyapunov/Sylvester equations 
 \begin{align}\label{eq:Lyapunov}
U_1&= L^\ast R_S^{-1}L + F_S^\ast U_1 F_S \nn\\
U_2&= G_1G_1^\ast + F_S U_2 F_W^\ast \nn\\
U_3&= F_EU_3F_A^\ast+ G_E\tilde{H}_A
\end{align}
with \begin{equation}
    \tilde{H}_A = G_2^\ast \begin{pmatrix}
        U_1U_2&F_S^\ast 
    \end{pmatrix},
\end{equation} and the constants for the anti-causal function are
\begin{align}
    H_A &= - R_M^{1/2} (\tilde{H}_A +  K_M U_3F_A^\ast), \nn\\
    F_A^\ast&= \begin{pmatrix}
        F_W^\ast & 0\\
        U_1G_1G_1^\ast & F_S^\ast
    \end{pmatrix}, \quad G_A = \begin{pmatrix}
    H^\ast R_W^{-\ast/2}\\
    0 
    \end{pmatrix}.
\end{align}

\end{lemma}

\subsection{Final Solution}\label{subsec:finalsol}
In Lemma \ref{lemma:nehariSol}, we employ Theorem 5 from \cite{sabag2021regret} to determine the optimal causal operator $\mathcal C_N$. This operator serves as the solution to the Nehari problem \eqref{eq:th_reduction_Nehari}, which can also be expressed in terms of transfer matrices as follows:
\begin{align}\label{eq:Nehari_tf}
    \inf_{\text{causal} \ C(z)} \|  C(z) - \{X(z)\}_- \|\le 1,
\end{align}
with  
\begin{equation}
    \{X(z)\}_{-}=A(z) = H_A\a{F_A^\ast}G_A.
\end{equation}
\begin{lemma}\label{lemma:nehariSol}
The solution to the Nehari problem \eqref{eq:Nehari_tf} is 
\begin{align}\label{eq::CN}
    C_N(z)&= H_A \Pi(F_N \c{F_N} +I) K_N,
\end{align}
where
\begin{align}
    K_N &= (I - F_A Z F_A^\ast\Pi)^{-1}F_A Z G_A, \quad F_N = F_A - K_NG_A^\ast,
\end{align}
and $Z$ and $\Pi$ are the solutions to the Lyapunov equations
\begin{align}
    Z&= F_A Z F_A^\ast + H_A^\ast H_A\nn \\
    \Pi&= F_A^\ast \Pi F_A + G_AG_A^\ast.
\end{align}
\end{lemma}
 Recall from Theorem \ref{th:reduction} that the Youla parameterization of the causal regret-optimal controller is $\mathcal Q_c $ \eqref{eq::optCont} and thus its transfer matrix can be written as
\begin{align}\label{eq:newQc}
    Q_{c}(z)&= M^{-/2}(z)(\{X(z)\}_{+} + C_N(z))W^{1/2}(z)\nn\\
    &=M^{-/2}(z)(C_1(z)+C_2(z) + C_N(z))W^{1/2}(z)
\end{align}
with $C_1(z)$ \eqref{eq::C1}, $C_2(z)$ \eqref{eq::C2}  are as found by Lemma \ref{lemma:decomposition}, and $C_N(z)$ \eqref{eq::CN} is the solution to the Nehari problem as found by Lemma \ref{lemma:nehariSol}. $M^{-/2}(z)$ \eqref{eq::thetaInv} is as found by Lemma \ref{lemma:factor_M2}, and $W^{1/2}(z)$ \eqref{eq::Whalf} is as found by Lemma \ref{lemma:factor_W}. 

With these components, we can now present our main result: in Theorem \ref{th:SS_causal}, we establish a state-space formulation for the regret-optimal controller $Q_c$ \eqref{eq:newQc} which, according to Theorem \ref{th:reduction}, solves Problem \ref{pb::RO}:
\begin{theorem}\label{th:SS_causal}
Under the condition of Assumption \ref{ass:stab}, the causal regret-optimal controller corresponding to $Q_c(z)$ \eqref{eq:newQc} which solves Problem \ref{pb::RO} suboptimally is:
\begin{align}
    \xi_{i+1} &= A \xi_i + B y_i\nn \\
    q_i  &= C \xi_i + D y_i
\end{align}
where:
\begin{align}
    A&= \begin{pmatrix}
    F_W & 0 & 0 &0 \\
    -K_NR_W^{-/2}H & F_N&0&0 \\
    - U_3 G_AR_W^{-/2}H & G_ER_M^{-/2}H_A \Pi& F_M&0 \\
    -U_2 H^\ast R_W^{-1}H & 0& 0& F_S
    \end{pmatrix}\nn\\
    B &= \begin{pmatrix}
     F_WK_W\\
     -K_NR_W^{-/2}H K_W+  F_N K_NR_W^{-/2}\\
B_3\\
      -U_2 H^\ast R_W^{-1}H K_W + F_S U_2 H^\ast R_W^{-1}
    \end{pmatrix} \nn \\
        C &=     \begin{pmatrix}
    0& R_{M}^{-/2}H_A\Pi& - K_M & -G_2^\ast U_1 F_S
    \end{pmatrix}\nn\\
    D&= [R_{M}^{-/2}H_A\Pi K_N   - K_M U_3 G_A \nn \\ & \ \ \ -G_2^\ast U_1 F_S U_2 H^\ast R_W^{-\ast/2}] R_W^{-/2}.
\end{align}
with 
\begin{align}
  B_3 &\triangleq - U_3 G_AR_W^{-/2}H K_W +  G_ER_M^{-/2}H_A \Pi K_NR_W^{-/2} \nn\\
&\ \ \ +  F_M U_3 G_AR_W^{-/2}  
\end{align}
\end{theorem}
This theorem shows that we are able to find an online controller that minimizes regret over an infinite-horizon in a partially observable system. It allows for practical implementation of the controller as we will see in Section \ref{sec::sim}.
\section{Simulations}\label{sec::sim}
 In this section, we showcase the the practicality of implementing the regret-optimal controller and its effectiveness as compared to the $\mathcal H_2$ and $\mathcal H_\infty$ controllers, for various systems. We assess them in the frequency domain to enable a comprehensive comparison across a wide spectrum of disturbances.

 We examine four distinct systems: 1) A system with $n=m=p=q=4$ featuring a randomly generated unstable matrix $F$, with all other matrices set to $I$, 2) a system with $n=m=p=q=2$ characterized by a randomly generated marginally stable matrix $F$, while other matrices set to $I$, 3) an aircraft dynamic model (AC15) \cite{aircraft} with $F$  unstable and 4) an aircraft dynamic model (AC5) with $F$ marginally stable as described in \cite{aircraft}. In our analysis, we depict three key performance metrics: i) The squared Frobenius norm minimized by the $\mathcal H_2$ controller \eqref{eq::frobnorm}, ii) the squared operator norm minimized by $\mathcal H_\infty$ controller \eqref{eq::opnorm}, and iii) the regret minimized by the regret-optimal controller \eqref{eq::regnorm}, as is shown in the equations below:
 \begin{align}
     &\|T_Q\|_F^2=\frac{1}{2\pi}\int_0^{2\pi}Tr(T_Q^\ast(e^{jw})T_Q(e^{jw})) dw,\label{eq::frobnorm}  \\
     &\|T_Q\|^2=\max_{0\leq w \leq 2\pi} \sigma_{max}(T_Q^\ast(e^{jw})T_Q(e^{jw})), \label{eq::opnorm}\\
     &\|T_Q^\ast T_Q-T_{Q_0}^\ast T_{Q_0}\|=\max_{0\leq w \leq 2\pi} \sigma_{max}(T_Q^\ast(e^{jw})T_Q(e^{jw})-\nn \\ &\quad \quad \quad \quad \quad \quad \quad \quad \quad \quad \quad \quad \quad \quad \quad T_{Q_0}^\ast(e^{jw})T_{Q_0}(e^{jw})).\label{eq::regnorm}
 \end{align}

 The performance metrics across all frequencies are shown in Figure \ref{fig:performance} for the (AC15) system. Table \ref{table:main}
shows the different norms \eqref{eq::frobnorm}-\eqref{eq::regnorm} resulting for each controller for each of the four systems.

The non-causal controller achieves minimal Frobenius and operator norms. Specifically, $H_2$ minimizes the Frobenius norm (area under the curve) but exhibits notable peaks at certain frequencies (Figure \ref{fig:frob}, Table \ref{table:main}), while $H_\infty$ minimizes the operator norm (Figure \ref{fig:op}, Table \ref{table:main}) but demonstrates subpar average performance (large area under the curve in Figure \ref{fig:frob}). In contrast, the regret-optimal controller combines attributes of both $H_2$ and $H_\infty$. Its Frobenius and operator norms interpolate between those of $H_2$ and $H_\infty$, and it boasts the lowest regret (Figure \ref{fig:regret}, Table \ref{table:main}). Furthermore, it consistently maintains a stable separation from the non-causal controller, distinguishing itself from $H_2$ and $H_\infty$ controllers. By minimizing the maximum cost deviation from the non-causal controller, it achieves balanced performance across all input disturbances.

\begin{table}
    \centering
    \caption{Controllers' performance across different systems with highlighted best metrics}
    \label{tab:combined}
    \begin{subtable}{\linewidth}
        \centering
        \caption{Controllers' performance on random systems}
    \label{tab:performance}
    \begin{tabular}{|l|c|c|c||c|c|c|}
        \hline 
        & \multicolumn{3}{c||}{\shortstack{System 1 (random,\\unstable $F$)}} & \multicolumn{3}{c|}{\shortstack{System 2 (random, \\marginally stable $F$)}}  \\
        \hline \hline
        & $\|T_Q\|_F^2$ & $\|T_Q\|^2$ & Regret & $\|T_Q\|_F^2$ & $\|T_Q\|^2$ & Regret \\
        \hline
        NC &2.974 & 1.1371&-&1.4338 & 1.1366 & -\\
        \hline
        RO &6.7121 &4.0726 &\textbf{3.1322} &6.1077 &3.7573 &\textbf{2.7821} \\
        \hline
        H2 &\textbf{5.6905} &7.8981 & 6.9475&\textbf{4.5032} &6.8703 & 5.8715\\
        \hline
        H$\infty$ & 6.9961&  \textbf{3.6896}& 3.6365& 6.6593&\textbf{3.3806} & 3.2927\\
        \hline
    \end{tabular}
    \end{subtable}
    \quad % Add some space between the subtables
    \\
    \begin{subtable}{\linewidth}
        \centering
          \caption{Controllers' performance on aircraft models \cite{aircraft}}
    \begin{tabular}{|l|c|c|c||c|c|c|}
        \hline 
        & \multicolumn{3}{c||}{System 3 (AC5)} & \multicolumn{3}{c|}{System 4 (AC15)}  \\
        \hline \hline
        & $\|T_Q\|_F^2$ & $\|T_Q\|^2$ & Regret & $\|T_Q\|_F^2$ & $\|T_Q\|^2$ & Regret \\
        \hline
        NC & 41.060&354.28& -&66.614 & 222.93&- \\
        \hline
        RO &1048.1 &1277.1 & \textbf{1006.8}&434.21 & 577.79&\textbf{363.24} \\
        \hline
        H2 & \textbf{269.03}& 2858.3&  2663.8&\textbf{343.34} &1054.2 &841.49 \\
        \hline
        H$\infty$ & 1187.1& \textbf{1166.1}& 1166.1&497.36 & \textbf{477.22}& 470.57\\
        \hline
    \end{tabular}
    \end{subtable}
    \label{table:main}
\end{table}

\begin{figure}
    \centering
    \begin{subfigure}{0.48\columnwidth}
        \includegraphics[width=\linewidth]{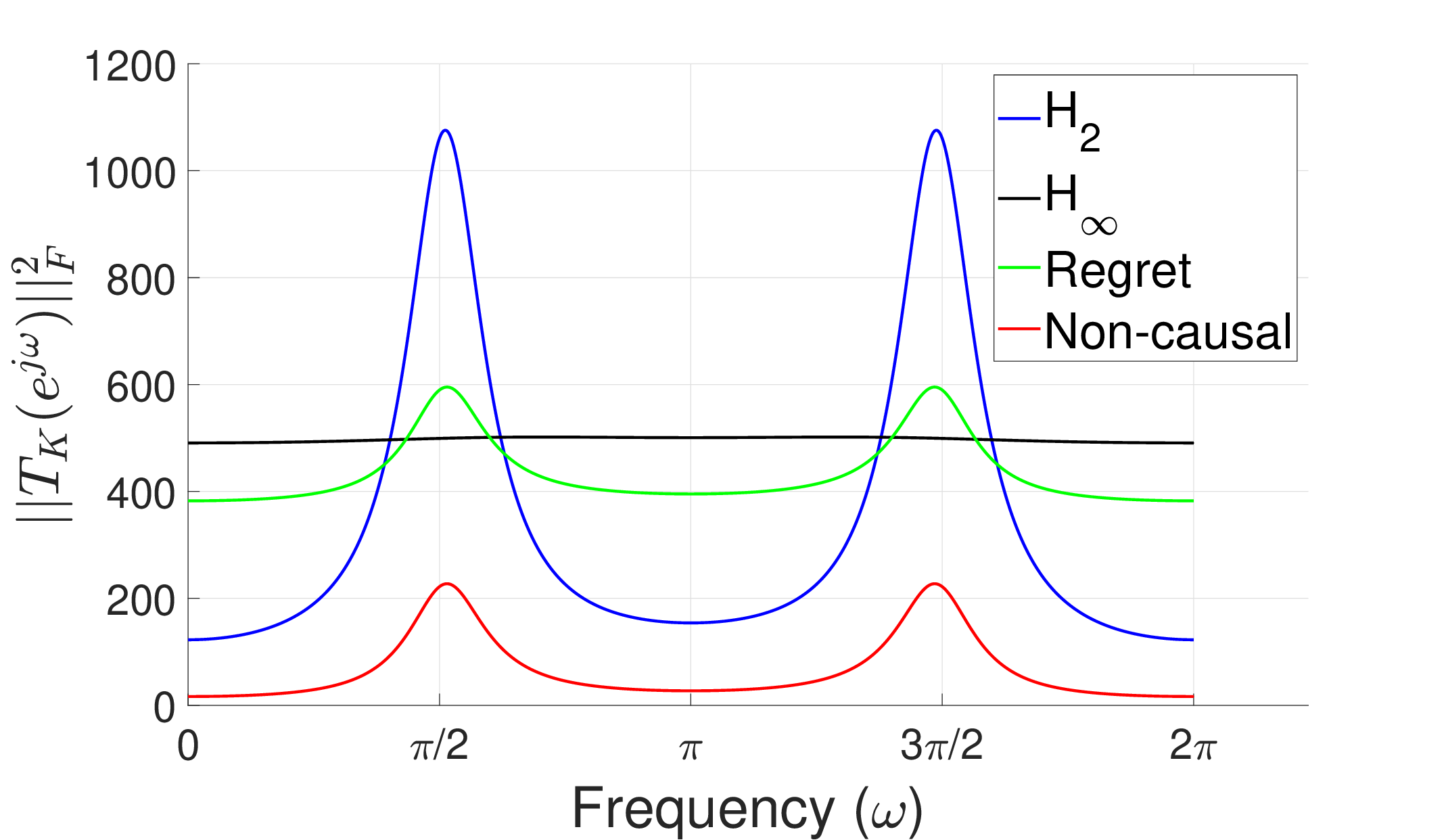}
        \caption{Squared Frobenius norm}
        \label{fig:frob}
    \end{subfigure}
    \hfill
    \begin{subfigure}{0.48\columnwidth}
        \includegraphics[width=\linewidth]{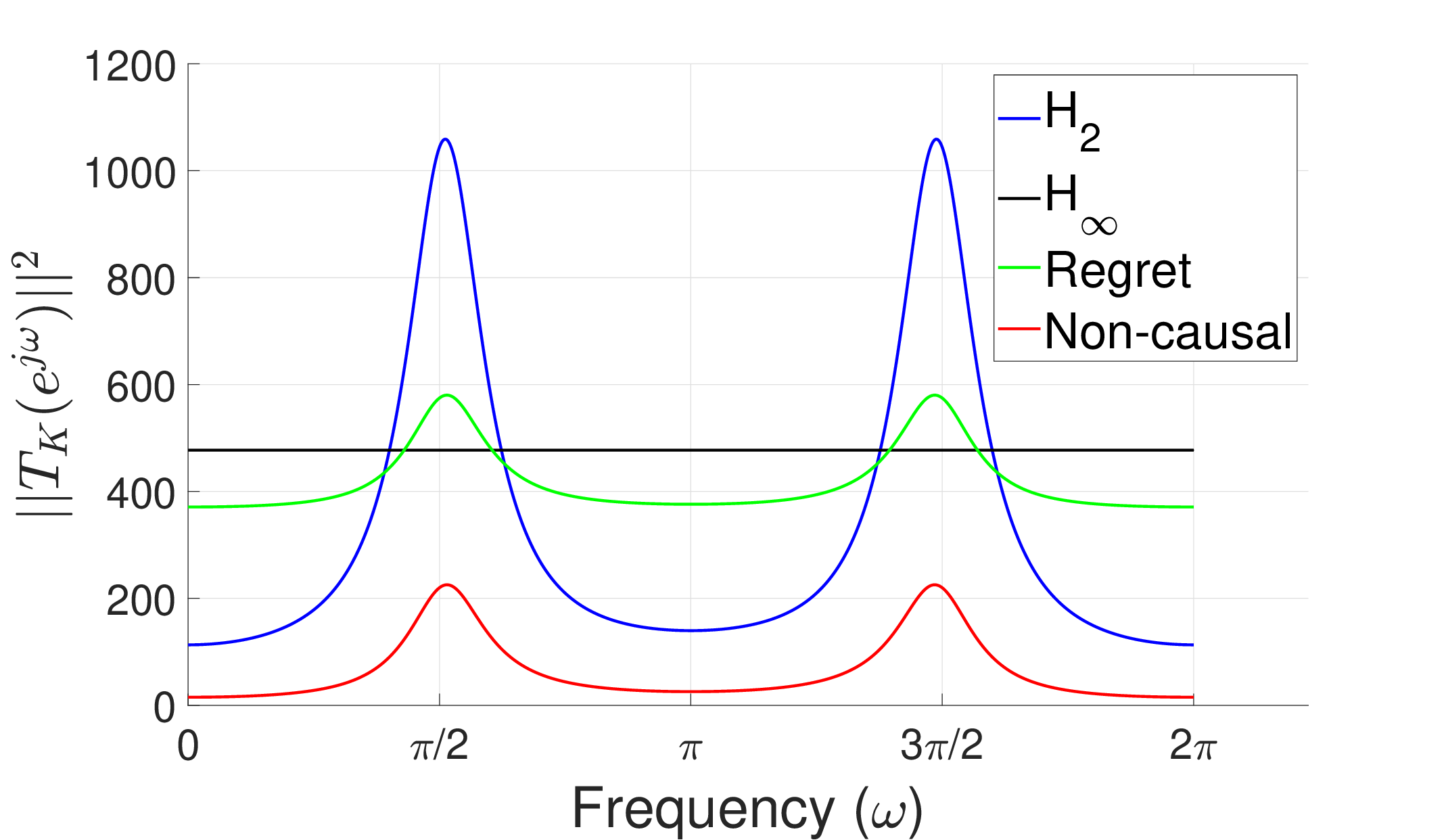}
        \caption{Squared operator norm}
        \label{fig:op}
    \end{subfigure}
    
    \vspace{\baselineskip} % Add some vertical space between the rows of images
    
    \begin{subfigure}{\columnwidth}
        \centering
        \includegraphics[width=0.48\linewidth]{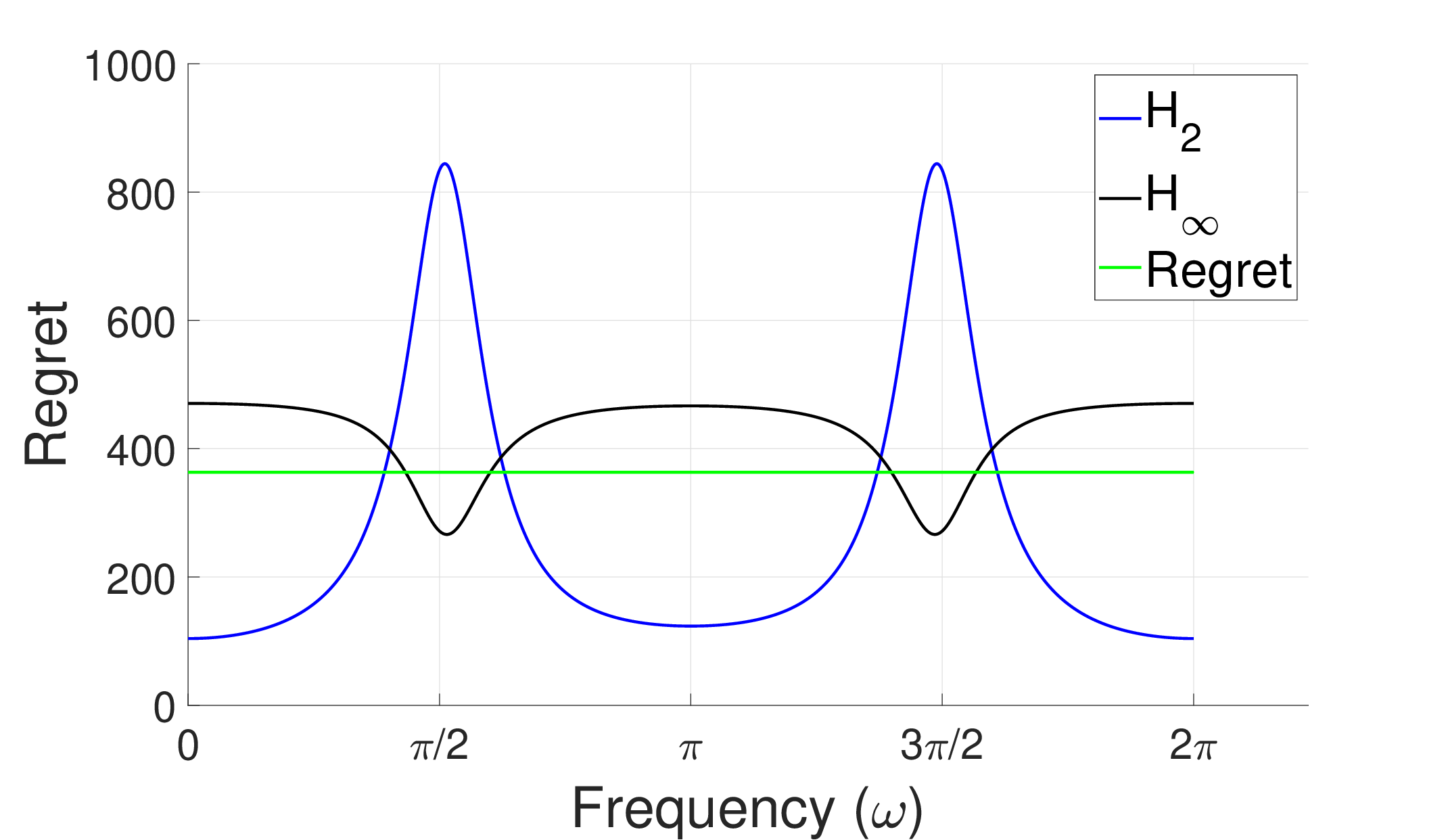}
        \caption{Regret}
        \label{fig:regret}
    \end{subfigure}
    \caption{Performance metrics for AC15 \cite{aircraft}}
    \label{fig:performance}
\end{figure}
\bibliographystyle{./bibliography/IEEEtran}
\bibliography{./bibliography/IEEEabrv,./bibliography/IEEEexample} 
\appendices
\section{Canonical Factorizations}
In this appendix, we introduce two key lemmas illustrating the canonical factorization of a causal operator. In Lemma \ref{lemma:aux1_factor}, we demonstrate the anticausal/causal or upper/lower (UL) factorization, while in Lemma \ref{lemma:aux2_factor}, we showcase the causal/anticausal or lower/upper (LU) factorization.
\begin{lemma}\label{lemma:aux1_factor}
Consider a general strictly causal transfer matrix $\underline{M}(z) = \underline{H} \c{\underline{F}} \underline{G}$ with $\underline{R}\succeq0$ and ($\underline{F},\underline{G}$) stabilizable. Then, the following canonical factorization holds
\begin{align}
    \Omega^\ast(z^{-\ast})\Omega(z)&= R + \underline{M}^\ast(z^{-\ast})\underline{M}(z)
\end{align}
with the causal transfer matrix
\begin{align}
    \Omega(z)&= R_P^{1/2} (I + K_P \c{\underline{F}}\underline{G}),
\end{align}
where 
\begin{align}
    K_P&= R_P^{-1}\underline{G}^\ast P\underline{F}\\
    R_P&= R + \underline{G}^\ast P\underline{G} = R_P^{\ast/2}R_P^{1/2},
\end{align}
and $P=P^\ast$ is the unique stabilizing solution to the Riccati equation
\begin{align}
    P&= \underline{F}^\ast P \underline{F}   + \underline{H}^\ast \underline{H} - K_P^\ast R_P K_P.
\end{align}

The inverse of the factorization is causal and bounded on the unit circle, and is equal to
    \begin{align}\label{eq::omegInv}
        \Omega^{-1}(z)&=  (I - K_P \c{\underline{F}_P}\underline{G}) R_P^{-/2} 
    \end{align}
    with \begin{equation}
        F_P= \underline{F} - \underline{G}K_P.
    \end{equation}

\end{lemma}

\begin{proof}[Proof of Lemma \ref{lemma:aux1_factor}]
We can express $R+ \underline{M}^\ast(z^{-\ast})\underline{M}(z)$ in matrix form as
\begin{align}\label{eq:proof_factorization}
&\begin{pmatrix} \underline{G}^{*}\a{\underline{F}^\ast} & I \end{pmatrix} \begin{pmatrix} \underline{H}^\ast \underline{H} & 0 \\ 0 & R\end{pmatrix} \begin{pmatrix} \c{\underline{F}} \underline{G} \\ I \end{pmatrix}\nn\\
&\ =\begin{pmatrix} \underline{G}^{*} \a{\underline{F}^\ast} & I \end{pmatrix}\times \nn \\ &\ \ \ \begin{pmatrix} \underline{H}^\ast \underline{H} - P + \underline{F}^{*}P\underline{F} & \underline{F}^{*}P \underline{G} \\ \underline{G}^\ast P \underline{F} & R + \underline{G}^\ast P \underline{G} \end{pmatrix} \begin{pmatrix} \c{\underline{F}}\underline{G}  \\ I \end{pmatrix},
\end{align}
where the equality holds for any $P = P^\ast$ from the identity
\begin{align}
   0&= \begin{pmatrix} \underline{G}^{*} \a{\underline{F}^\ast} & I \end{pmatrix} \begin{pmatrix} - P + \underline{F}^{*}P\underline{F} & \underline{F}^{*}P \underline{G} \\ \underline{G}^\ast P \underline{F} & \underline{G}^\ast P \underline{G} \end{pmatrix} \times\nn \\ &\ \ \ \  \begin{pmatrix} \c{\underline{F}}\underline{G}  \\ I \end{pmatrix}.
\end{align}

The middle matrix in \eqref{eq:proof_factorization} can be factored as \begin{equation}
    \begin{pmatrix} I & K_P^\ast \\ 0 & I \end{pmatrix} \begin{pmatrix} \Gamma(P) & 0   \\0 & R_P  \end{pmatrix} \begin{pmatrix} I & 0 \\ K_P & I \end{pmatrix},
\end{equation} with 
\begin{align}
\Gamma(P) &= - P + \underline{F}^{*}P\underline{F} + \underline{H}^\ast \underline{H} - K_P^\ast R_PK_P\nn\\
K_P &= R_P^{-1}\underline{G}^{*}P\underline{F}\nn\\
R_P     &= R + \underline{G}^\ast P\underline{G}.
\end{align}
Since $(\underline{F}, \underline{G})$ is stabilizable, then the Riccati equation $\Gamma(P) = 0$ has a unique stabilizing Hermitian solution. Suppose $P$ is chosen to be this solution. Then, we can obtain the factorization by defining
\begin{equation}
    \Omega(z) =  R_P^{1/2} (I + K_P \c{\underline{F}}\underline{G})
\end{equation}
with $R_P = R_P^{\ast/2}R_P^{1/2}$.

Moreover, $\Omega^{-1}(z)$ \eqref{eq::omegInv} is found using the matrix inversion lemma.  Its poles are at the eigenvalues of $F_P$ which is stable (since $P$ is the stabilizing solution to the Riccati equation), which due to the causality of $\Omega^{-1}(z)$ guarantees the boundedness of $\Omega^{-1}(z)$ on the unit circle.
\end{proof}
\begin{lemma}\label{lemma:aux2_factor}
Consider a general strictly causal transfer matrix $\underline{M}(z) = \underline{H} \c{\underline{F}} \underline{G}$ with $R\succeq0$ and ($\underline{F},\underline{H}$) detectable. Then, the following canonical factorization holds
\begin{align}
    \Omega(z)\Omega^\ast(z^{-\ast})&= R + \underline{M}(z)\underline{M}^\ast(z^{-\ast})
\end{align}
with the causal transfer matrix
\begin{align}\label{eq:app:omega}
    \Omega(z)&= (I + \underline{H} \c{F}K_P) R_P^{1/2},
\end{align}
where 
\begin{align}
    K_P&= \underline{F}P\underline{H}^\ast R_P^{-1} \label{eq:app:K}\\
    R_P&= R + \underline{H} P\underline{H}^\ast = R_P^{1/2}R_P^{\ast/2},\label{eq:app:R}
\end{align}
and $P=P^\ast$ is the unique stabilizing solution to the Riccati equation
\begin{align}
    P&= \underline{F} P \underline{F}^\ast + \underline{G}\underline{G}^\ast - K_P R_P K_P^\ast.
\end{align}
The inverse of the factorization is causal and bounded on the unit circle, and is equal to
    \begin{align}
        \Omega^{-1}(z)&= R_P^{-/2} (I - \underline{H} \c{F_P}K_P) \label{eq:app:inv} 
    \end{align}
    with \begin{equation}
        F_P= \underline{F} - K_PH.\label{eq:app:F}
    \end{equation}
\end{lemma}
\begin{proof}[Proof of Lemma \ref{lemma:aux2_factor}]
The proof is similar to the proof of Lemma \ref{lemma:aux1_factor} and will be ommitted here. Refer to Theorem 8.3.1 in \cite{LEbabak} for a detailed proof. 
\end{proof}

\noindent\textbf{Corollary to Lemma \ref{lemma:aux2_factor}. }
\textit{Consider a general strictly causal transfer matrix $\underline{M}(z) = \underline{H} \c{\underline{F}} \underline{G}$ with $R\succeq0$ and ($\underline{F},\underline{H}$) detectable. Then, the following canonical factorization holds
\begin{align}  \Omega(z)\Omega^\ast(z^{-\ast})&= R - \underline{M}(z)\underline{M}^\ast(z^{-\ast})
\end{align}
with the causal transfer matrix $ \Omega(z)$ as in \eqref{eq:app:omega}, $K_P$ as in \eqref{eq:app:K}, $R_P$ as in \eqref{eq:app:R}, 
and $P=P^\ast$ is the unique stabilizing solution to the Riccati equation
\begin{align}
    P&= \underline{F} P \underline{F}^\ast - \underline{G}\underline{G}^\ast - K_P R_P K_P^\ast.
\end{align}
The inverse of the factorization is causal and bounded on the unit circle, and is equal to $\Omega^{-1}(z)$ as in \eqref{eq:app:inv}, and $F_P$ as in \eqref{eq:app:F}.}

\section{Proofs}\label{sec:proofs}
In this appendix, we present the proofs for the theorems and lemmas presented in the paper.
\begin{proof}[Proof of Theorem \ref{th:reduction}]
To reduce the sub-optimal regret problem to a Nehari problem with respect to the reference controller $\mathcal Q_0$, we should simplify the inequality
\begin{align}\label{eq:proof_reduction_regret_def}
T_{\mathcal Q}^\ast  T_{\mathcal Q} - T_{\mathcal Q_0} T_{\mathcal Q_0} \preceq \gamma^2\mathcal I,
\end{align}
which is equivalent to 
\begin{align}\label{eq:proof_reduction_subopt_rotate}
        (\theta T_{\mathcal Q} \psi)^\ast \theta T_{\mathcal Q} \psi - (\theta T_{\mathcal Q_0} \psi)^\ast \theta T_{\mathcal Q_0} \psi \preceq \gamma^2 \mathcal I,
\end{align}
for any unitary operators $\theta$ and $\psi$. We choose the unitary operators 
\begin{align}
\theta &= 
% \begin{pmatrix}
% \mathcal S^{-/2}& -\mathcal S^{-/2} \mathcal P_{12}\\
% \mathcal T^{-\ast/2} \mathcal P_{12}^\ast&\mathcal T^{-\ast/2} 
% \end{pmatrix} 
 \begin{pmatrix}
\mathcal S^{-/2}& 0\\
0&\mathcal T^{-\ast/2} 
\end{pmatrix}
\begin{pmatrix}
I & - \mathcal P_{12}\\
 \mathcal P_{12}^\ast& I 
\end{pmatrix}\nn\\
\psi 
&= \begin{pmatrix}
I& \mathcal P_{21}^\ast\\
-\mathcal P_{21}&I
\end{pmatrix} \begin{pmatrix}
\mathcal V^{-/2}&0\\
0&\mathcal U^{-\ast/2}
\end{pmatrix},
\end{align}
where the canonical factorizations are defined as 
\begin{align}
\mathcal S^{1/2} \mathcal S^{\ast/2}&= I + \mathcal P_{12} \mathcal P_{12}^\ast , \ \ \ \mathcal U^{1/2}\mathcal U^{\ast/2} = I + \mathcal P_{21} \mathcal P_{21}^\ast, \nn\\
\mathcal T^{\ast/2} \mathcal T^{1/2}&= I + \mathcal P_{12}^\ast \mathcal P_{12}, \ \ \ \mathcal V^{\ast/2} \mathcal V^{1/2}= I + \mathcal P_{21}^\ast \mathcal P_{21},
\end{align}
with $\mathcal E^{1/2}$ being a causal operator for any positive operator~$\mathcal E$.

The operator $\theta T_{\mathcal Q} \psi$ can be computed from the product of its three middle matrices as
\begin{align}\label{eq:proof_reduction_prod3}
&\begin{pmatrix}
I& - \mathcal P_{12}\\
\mathcal P_{12}^\ast&I
\end{pmatrix} 
\begin{pmatrix}
    \mathcal P_{12} \mathcal Q \mathcal P_{21} + \mathcal P_{11}& \mathcal P_{12} \mathcal Q\\
    \mathcal Q \mathcal P_{21}& \mathcal Q
    \end{pmatrix} \times \nn
   \\
   &  \begin{pmatrix}
I& \mathcal P_{21}^\ast\\
-\mathcal P_{21}&I
\end{pmatrix}\nn \\
&=\nn\\
& \begin{pmatrix}
\mathcal P_{11} & \mathcal P_{11}\mathcal P_{21}^\ast\\
\mathcal P_{12}^\ast \mathcal P_{11}& ( I + \mathcal P_{12}^\ast \mathcal P_{12})\mathcal Q (I + \mathcal P_{21}\mathcal P_{21}^\ast) + \mathcal P_{12}^\ast \mathcal P_{11}\mathcal P_{21}^\ast
\end{pmatrix}. \nn
\end{align}
The remaining diagonal matrices can be easily added to obtain $\theta T_{\mathcal Q} \psi$. 
% \begin{align}
%     \theta T_\mathcal{K}\psi&= \begin{pmatrix}
% S^{-/2}&0\\
% 0&T^{-\ast/2}
% \end{pmatrix} \begin{pmatrix}
% G& GL^\ast\\
% F^\ast G& F^\ast GL^\ast + ( I + F^\ast F)Q (I + LL^\ast)  
% \end{pmatrix}\begin{pmatrix}
% V^{-/2}&0\\
% 0&U^{-\ast/2}
% \end{pmatrix}.
% \end{align}
It also follows directly that the non-causal $H_2$ controller is 
\begin{equation}
    \mathcal Q_2 \triangleq - ( I + \mathcal P_{12}^\ast \mathcal P_{12})^{-1}\mathcal P_{12}^\ast \mathcal P_{11}\mathcal P_{21}^\ast(I + \mathcal P_{21}\mathcal P_{21}^\ast)^{-1}
\end{equation}
since the Frobenius norm is minimized when the $(2,2)$ term is equal to zero. Now, omitting additive terms that do not depend on $\mathcal Q$ (since they have no effect on the regret), the squared cost operator  is
\begin{align}
(\theta T_Q \psi)^\ast \theta T_Q \psi &\dot{=}\nn \\ 
   &\begin{pmatrix}
    0   &  
    \mathcal V^{-\ast/2} \mathcal P_{11}^\ast \mathcal P_{12} ( \mathcal Q - \mathcal Q_2)\mathcal U^{1/2}\\
    (\cdot)^\ast & \mathcal U^{\ast/2}( \mathcal Q^\ast - \mathcal Q_2^\ast )\mathcal T( \mathcal Q - \mathcal Q_2 )\mathcal U^{1/2}
    \end{pmatrix}.
    %         &= 
    % \begin{pmatrix}U^{1/2}
    % \end{pmatrix}
\end{align}
The sub-optimal regret problem in \eqref{eq:proof_reduction_subopt_rotate} can be written as
    \begin{align}\label{eq:proof_reduction_regret_explicit}
    \begin{pmatrix}
    \gamma^2 \mathcal I &  \mathcal V^{-\ast/2} \mathcal P_{11}^\ast \mathcal P_{12} 
    (\mathcal Q_0 - \mathcal Q)\mathcal U^{1/2} \\
    (\cdot)^\ast & \mathcal N_{2,2}
    \end{pmatrix}\succeq 0,
\end{align}
where \begin{align}
    \mathcal N_{2,2}&\triangleq \gamma^2 \mathcal I + \mathcal U^{\ast/2}( \mathcal Q_0^\ast - \mathcal Q_2^\ast )\mathcal T ( \mathcal Q_0 -\mathcal Q_2 )\mathcal U^{1/2} - \mathcal U^{\ast/2}\times \nn \\&\ \ \ \ ( \mathcal Q^\ast - \mathcal Q_2^\ast )\mathcal T( \mathcal Q - \mathcal Q_2 )\mathcal U^{1/2}.
\end{align}
Since $\gamma^2 \mathcal I\succ0$, the condition \eqref{eq:proof_reduction_regret_explicit} can be rewritten using its Schur complement as
\begin{align}\label{eq:proof_reduction_PI}
    \mathcal N_{2,2} &\succeq  \mathcal U^{\ast/2} (\mathcal Q - \mathcal Q_0)^\ast \mathcal P_{12}^\ast \mathcal P_{11}
    \gamma^{-2} \mathcal V^{-1} \mathcal P_{11}^\ast \mathcal P_{12}
    \times \nn \\ &\ \ \ \ (\mathcal Q - \mathcal Q_0)\mathcal U^{1/2}.
\end{align}
Using the completion of squares, we can re-write \eqref{eq:proof_reduction_PI} as 
\begin{align}\label{eq:proof_reduction_2}
    &(\mathcal Q^\ast - \mathcal B^\ast \mathcal D^{-1}) \mathcal D (\mathcal Q - \mathcal D^{-1}\mathcal B)\nn\\
    &\ \ \ \preceq \gamma^2 \mathcal U^{-1} + (\mathcal Q_2 - \mathcal Q_0)^\ast \mathcal T \mathcal D^{-1}\mathcal T(\mathcal Q_2 - \mathcal Q_0)
\end{align}
with 
\begin{equation}
    \mathcal D \triangleq \mathcal T + \mathcal P_{12}^\ast \mathcal P_{11}
    \gamma^{-2} \mathcal V^{-1} \mathcal P_{11}^\ast \mathcal P_{12}
\end{equation}
and \begin{align}
    \mathcal B &\triangleq \mathcal T \mathcal  Q_2  + \mathcal P_{12}^\ast \mathcal P_{11}
    \gamma^{-2} \mathcal V^{-1} \mathcal P_{11}^\ast \mathcal P_{12}\mathcal Q_0 \nn\\&= \mathcal T (\mathcal  Q_2 - \mathcal Q_0)  + \mathcal D \mathcal Q_0.
\end{align}

To complete the proof, define the canonical factorizations of the positive operators 
\begin{align}\label{eq:proof_reduction_MN}
     \mathcal W&= \mathcal U^{-1} + \gamma^{-4} (\mathcal Q_2 - \mathcal Q_0)^\ast \mathcal T \mathcal M^{-1}\mathcal T(\mathcal Q_2 - \mathcal Q_0)\\
     \mathcal M&= \gamma^{-2}\mathcal I+ \gamma^{-2} \mathcal P_{12}^\ast(\mathcal I +  \gamma^{-2}\mathcal P_{11}(\mathcal I + \mathcal P_{21}^\ast \mathcal P_{21})^{-1} \mathcal P_{11}^\ast)\times \nn \\ &\ \ \ \ \mathcal P_{12},
\end{align}
to express concisely the condition \eqref{eq:proof_reduction_2} as 
\begin{align}
(\mathcal M^{1/2}\mathcal Q\mathcal W^{-/2} - \mathcal X  )^\ast  (\mathcal M^{1/2}\mathcal Q\mathcal W^{-/2} - \mathcal X ) \preceq \mathcal I 
\end{align}
with 
\begin{equation}
    \mathcal X \triangleq \mathcal M^{1/2}\mathcal Q_0\mathcal W^{-/2} + \gamma^{-2}\mathcal M^{-\ast/2} \mathcal T (\mathcal  Q_2 - \mathcal Q_0)\mathcal W^{-/2}.
\end{equation}
By decomposing the operator $\mathcal X = \{\mathcal X\}_+ + \{\mathcal X\}_-$, and the fact that the operator $\mathcal M^{1/2}{\mathcal Q}\mathcal W^{-/2}$ is a product of three causal operators, we conclude that \eqref{eq:proof_reduction_MN} is a Neahri problem of finding a causal operator to approximate the anticausal operator $\{\mathcal X\}_-$. Furthermore, if there exists a solution to this Nehari problem, say $\mathcal C$, then the solution to the sub-optimal regret problem in \eqref{eq:proof_reduction_regret_def} is 
\begin{equation}
    \mathcal Q = \mathcal M^{-/2}(\mathcal C + \{\mathcal X\}_+)\mathcal W^{1/2}.
\end{equation}
\end{proof}

\begin{proof}[Proof of Lemma \ref{lemma:factor_W} ] To obtain the required factorization, we apply Lemma \ref{lemma:aux2_factor} with $\underline{R}=I$, $\underline{H}=H=$, $\underline F=F$, and $\underline{G}=G_1$.
\end{proof}

\begin{proof}[Proof of Lemma \ref{lemma:factor_M}]

The inverse of the positive inner operator $I + \gamma^{-2} \mathcal P_{11} (I + \mathcal P_{21}^\ast \mathcal P_{21})^{-1} \mathcal P_{11}^\ast$ can be factorized as
\begin{align}
    \nabla \nabla^\ast&= (I + \gamma^{-2} \mathcal P_{11} (I + \mathcal P_{21}^\ast \mathcal P_{21})^{-1} \mathcal P_{11}^\ast)^{-1}\nn\\
    &= I - \mathcal P_{11} ( \gamma^{2}(I + \mathcal P_{21}^\ast \mathcal P_{21}) + \mathcal P_{11}^\ast \mathcal P_{11})^{-1} \mathcal P_{11}^\ast\nn\\
    &= I - \mathcal P_{11} ( \gamma^{2} (I + \mathcal P_{21}^\ast \mathcal P_{21}) + \mathcal P_{11}^\ast \mathcal P_{11})^{-1}\mathcal P_{11}^\ast\nn\\
    &= I - \mathcal P_{11} \Gamma^{-1}\Gamma^{-\ast} \mathcal P_{11}^\ast,
\end{align}
where in the last equality we use the canonical factorization
\begin{align}
 \Gamma^\ast \Gamma&= \gamma^{2} (I +\mathcal P_{21}^\ast \mathcal P_{21}) + \mathcal P_{11}^\ast \mathcal P_{11}.
\end{align}
The operator can now be directly factorized as
\begin{equation}
    \mathcal M^{\ast/2} \mathcal M^{1/2} = \gamma^{-2}(\mathcal I + \mathcal P_{12}^\ast \nabla^{-\ast}\nabla^{-1} \mathcal P_{12}).
\end{equation}
\end{proof}
\begin{proof}[Proof of Lemma \ref{lemma:factor_gamma}]
We start by writing the transfer matrix explicitly
\begin{align}
    &\gamma^{2} I + \gamma^{2} P_{21}^\ast(z^{-\ast}) P_{21}(z) + P_{11}^\ast(z^{-\ast}) P_{11}(z) \nn\\
    &\ = \gamma^{2} I + G_1^\ast\a{F^\ast} (\gamma^{2} H^\ast H + L^\ast L)\times \nn \\ &\ \ \ \ \c{F}G_1.\nn
\end{align}
To obtain the required factorization, we apply Lemma \ref{lemma:aux1_factor} with $\underline R = \gamma^2I$, $\underline H = (\gamma^2 H^\ast H + L^\ast L)^{1/2}$, $\underline F=F$, and $\underline G = G_1$. 
\end{proof}

\begin{proof}[Proof of Lemma \ref{lemma:factor_nabla}]
We start by simplifying the product 
\begin{align}
    G(z) \Gamma^{-1}(z)
    %&= L\c{F}G_1(I + K_\Gamma \c{F}G_1)^{-1}R_\Gamma^{-/2}\nn\\
 %   &= L\c{F}(I + G_1 K_\Gamma \c{F})^{-1}G_1R_\Gamma^{-/2}\nn\\
    &= L\c{F_\Gamma}G_1R_\Gamma^{-/2}.
\end{align}
Thus, the factorization should be for the expression
\begin{align}
    I - L\c{F_\Gamma}G_1R_\Gamma^{-1}G_1^\ast \c{F_\Gamma^\ast} L^\ast.
\end{align}
We obtain the desired result by applying the Corollary to Lemma \ref{lemma:aux2_factor} with $\underline R = I$, $\underline H = L$, $\underline F = F_\Gamma$ and $\underline G = - G_1R_\Gamma^{-/2}$.
\end{proof}

\begin{proof}[Proof of Lemma \ref{lemma:factor_M2}]
We start by simplifying $\nabla^{-1}(z)P_{12}(z)$ as follows:
\begin{align}\label{eq:proof_M_simplified}
    &R_\nabla^{-/2}(I + L \c{F_\Gamma}K_\nabla)^{-1}   L\c{F}G_2\nn\\
    % &= R_\nabla^{-/2}L(I +  \c{F_\Gamma}K_\nabla L)^{-1}   \c{F}G_2\nn\\
    &= R_\nabla^{-/2}L(I - ( zI - F_\Gamma + K_\nabla L)^{-1}K_\nabla L)  \c{F}G_2\nn\\
    &= R_\nabla^{-/2}L(I - ( zI - F_\nabla)^{-1}K_\nabla L)  \c{F}G_2.
\end{align}
We can write \eqref{eq:proof_M_simplified} concisely as $$    \gamma^{-1}\nabla^{-1}(z)F(z)= H_E
    \c{F_E}G_E,$$
with 
\begin{align}
H_E &\triangleq \begin{pmatrix}
    R_\nabla^{-/2}L&R_\nabla^{-/2}L 
    \end{pmatrix} ,\nn \\ F_E &\triangleq \begin{pmatrix}
F&0\\
- K_\nabla L & F_\nabla
\end{pmatrix},  \quad G_E \triangleq 
    \begin{pmatrix}
    \gamma^{-1}G_2\\
    0 
    \end{pmatrix}.       
\end{align}

In order to complete the factorization of $M(z)$, we apply Lemma \ref{lemma:aux1_factor} with $\underline R = \gamma^{-2}I$, $\underline H=H_E$, $\underline F=F_E$, and $\underline G = G_E$.
\end{proof}

\begin{proof}[Proof of Lemma \ref{lemma:factorization_S}] The desired factorization can be obtained through applying Lemma \ref{lemma:aux2_factor} with $\underline R=I$, $\underline H=L$, $\underline F=F$, and $\underline G=G_2$.
\end{proof}
\begin{proof}[Proof of Lemma \ref{lemma:decomposition}]
The operator $$M^{1/2}(z)Q_2(z)W^{-/2}(z)$$ that needs to be decomposed can be expressed as
\begin{align}\label{eq:proof_dec_main}
  & - M^{1/2}(z) (I + P_{12}^\ast(z^{-\ast}) P_{12}(z))^{-1} P_{12}^\ast(z^{-\ast}) P_{11}(z) \times \nn\\
  &\ \ \ \ P_{21}^\ast(z^{-\ast}) (I + P_{21}(z)P_{21}^\ast(z^{-\ast}))^{-1} W^{-/2}(z) \nn \\
 &\stackrel{(a)}= - M^{1/2}(z) P_{12}^\ast(z^{-\ast}) S^{-\ast/2}(z^{-\ast})S^{-/2}(z)P_{11}(z)\times\nn \\ &\ \ \ \ P_{21}^\ast(z^{-\ast}) W^{\ast/2}(z^{-\ast}),
\end{align}
where $(a)$ follows from the factorization $$S^{1/2}(z)S^{\ast/2}(z^{-\ast}) = I + P_{12}(z)P_{12}^\ast(z^{-\ast}).$$ For convenience, the canonical factorizations required to compute \eqref{eq:proof_dec_main} (derived in Lemmas \ref{lemma:factor_W}, \ref{lemma:factor_M2}, and \ref{lemma:factorization_S}) are summarized as
\begin{align}
    M^{1/2}(z)&= R_M^{1/2} (I + K_M \c{F_E}G_E)\nn \\
    P_{12}^\ast(z^{-\ast})&= G_2^\ast \a{F^\ast}L^\ast \nn \\
    % T^{-/2}(z)&= (I + K_T \c{F} G_2)^{-1}R_T^{-/2}\\
    % T^{-\ast/2}(z)&= R_T^{-\ast/2}(I + G_2^\ast \a{F^\ast}K_T^\ast )^{-1}\\
    S^{-\ast/2}(z^{-\ast})&= (I +K_S^\ast \a{F^\ast}L^\ast)^{-1}R_S^{-\ast/2}\nn \\
    S^{-/2}(z)&= R_S^{-/2}(I + L \c{F}K_S)^{-1}\nn \\
    P_{11}(z)&=  L\c{F}G_1\nn \\
    P_{21}^\ast(z^{-\ast})&= G_1^\ast \a{F^\ast}H^\ast \nn \\
    %     V^{-/2}(z)&= (I + K_V \c{F} G_1)^{-1}R_V^{-/2}\\
    % V^{-\ast/2}(z)&= R_V^{-\ast/2}(I + G_1^\ast \a{F^\ast}K_V^\ast )^{-1}\\
    W^{\ast/2}(z^{-\ast})&= (I +K_W^\ast \a{F^\ast}H^\ast)^{-1}R_W^{-\ast/2}.
\end{align}
We begin by combining some of the pairs products in \eqref{eq:proof_dec_main}
\begin{align}
    &P_{12}^\ast(z^{-\ast})S^{-\ast/2}(z^{-\ast})\nn\\
    &= G_2^\ast \a{F^\ast} (I +L^\ast K_S^\ast \a{F^\ast})^{-1}L^\ast \times \nn \\ &\ \ \ \ R_S^{-\ast/2}\\
    &= G_2^\ast \a{F_S^\ast}L^\ast R_S^{-\ast/2},
\end{align}
and similarly
\begin{align}
    S^{-/2}(z)P_{11}(z) 
    % &= R_S^{-/2}(I + L \c{F}K_S)^{-1}L\c{F}G_1\nn \\
    % &= R_S^{-/2}L(I + \c{F}K_SL)^{-1}\c{F}G_1\nn \\
    &= R_S^{-/2}L\c{F_S}G_1\nn\\
    P_{21}^\ast(z^{-\ast})W^{\ast/2}(z^{-\ast})
    % &= G_1^\ast \a{F^\ast}H^\ast
    % (I +K_W^\ast \a{F^\ast}H^\ast)^{-1}R_W^{-\ast/2}\nn\\
    % &= G_1^\ast \a{F^\ast}(I +H^\ast K_W^\ast \a{F^\ast})^{-1}H^\ast R_W^{-\ast/2}\nn \\
    &= G_1^\ast \a{F_W^\ast}H^\ast R_W^{-\ast/2}.
\end{align}
The combination of these products can be decomposed as
\begin{align}\label{eq:proof_decompo_6product}
    &P_{12}^\ast(z^{-\ast}) S^{-\ast/2}(z^{-\ast})S^{-/2}(z) P_{11}(z) P_{21}^\ast(z^{-\ast}) W^{\ast/2}(z^{-\ast})\nn \\
    % &=G_2^\ast \a{F_S^\ast}L^\ast R_S^{-\ast/2}  R_S^{-/2}L\c{F_S}G_1G_1^\ast \a{F_W^\ast}H^\ast R_W^{-\ast/2}\nn \\
    &=G_2^\ast \a{F_S^\ast}L^\ast R_S^{-1}L\c{F_S} \times\nn\\
    &\ \ \ \ G_1G_1^\ast \a{F_W^\ast}H^\ast R_W^{-\ast/2}\nn \\
    &\stackrel{(a)}= G_2^\ast [F_S^\ast\a{F_S^\ast} U_1 + U_1F_S \c{F_S}   \nn\\
    &\ \ \ \ +U_1]G_1G_1^\ast \a{F_W^\ast}H^\ast R_W^{-\ast/2}\nn \\
    % &= G_2^\ast [\a{F_S^\ast}F_S^\ast U_1G_1G_1^\ast \a{F_W^\ast} + U_1F_S \c{F_S}G_1G_1^\ast \a{F_W^\ast} + U_1G_1G_1^\ast \a{F_W^\ast}]H^\ast R_W^{-\ast/2}\nn \\
    % &= G_2^\ast [\a{F_S^\ast}F_S^\ast U_1G_1G_1^\ast \a{F_W^\ast} +  U_1G_1G_1^\ast \a{F_W^\ast}]H^\ast R_W^{-\ast/2}\nn \\
    % & + G_2^\ast [U_1F_S \c{F_S}G_1G_1^\ast \a{F_W^\ast}]H^\ast R_W^{-\ast/2}\nn \\
    % &= G_2^\ast [\a{F_S^\ast}F_S^\ast U_1G_1G_1^\ast \a{F_W^\ast} +  U_1G_1G_1^\ast \a{F_W^\ast}]H^\ast R_W^{-\ast/2}\nn \\
    % & + G_2^\ast U_1F_S [\c{F_S}F_SU_2 + U_2 F_W^\ast \a{F_W^\ast} + U_2]H^\ast R_W^{-\ast/2}\nn \\
    &\stackrel{(b)}= G_2^\ast [F_S^\ast\a{F_S^\ast} U_1G_1G_1^\ast \a{F_W^\ast} \nn\\
    & \ \ \ \ + U_1G_1G_1^\ast \a{F_W^\ast} \nn \\ &\ \ \ \ + U_1F_S U_2 F_W^\ast \a{F_W^\ast}\nn \\
    & \ \ \ \ + U_1F_S\c{F_S}F_SU_2 + U_1F_SU_2]H^\ast R_W^{-\ast/2}
\end{align}
where $(a)$ follows from \begin{align}
    &\a{F_S^\ast}L^\ast R_S^{-1}L\c{F_S}\nn\\& = \a{F_S^\ast}F_S^\ast U_1 + U_1F_S \c{F_S} + U_1
\end{align} where $U_1$ is the solution to the Lyapunov equation \begin{equation}
    L^\ast R_S^{-1}L + F_S^\ast U_1 F_S = U_1,
\end{equation} and $(b)$ follows from \begin{align}
    &\c{F_S} G_1G_1^\ast \a{F_W^\ast} \nn \\ &= \c{F_S}F_sU_2 + U_2F_W^\ast \a{F_W^\ast}+U_2
\end{align} where $U_2$ is the solution to the Sylvester equation \begin{equation}G_1G_1^\ast + F_S U_2 F_W^\ast = U_2.\end{equation}

Before combining \eqref{eq:proof_decompo_6product} with $- M^{1/2}(z)$, we denote the strictly anticausal part of \eqref{eq:proof_decompo_6product} as
\begin{align}\label{eq:proof_decompo_anti}
&G_2^\ast[F_S^\ast\a{F_S^\ast} U_1G_1G_1^\ast \a{F_W^\ast} \nn\\
&\ \ \ +  U_1U_2 \a{F_W^\ast} ]H^\ast R_W^{-\ast/2} \nn\\
% &= G_2^\ast[F_S^\ast \a{F_S^\ast} U_1G_1G_1^\ast \a{F_W^\ast} +  U_1U_2 \a{F_W^\ast}] H^\ast R_W^{-\ast/2}\nn\\
&=     G_2^\ast \begin{pmatrix}
        U_1U_2&F_S^\ast 
    \end{pmatrix}
    \a{\begin{pmatrix}
        F_W^\ast & 0\\
        U_1G_1G_1^\ast & F_S^\ast
    \end{pmatrix}}\times\nn\\
    &\ \ \ \begin{pmatrix}
    H^\ast R_W^{-\ast/2}\\
    0 
    \end{pmatrix}\nn\\
    &\triangleq \tilde{H}_A \a{F_A^\ast }G_A.
\end{align}
We now decompose the product of $- M^{1/2}(z)$ and \eqref{eq:proof_decompo_anti} as
\begin{align}
    &- M^{1/2}(z)\tilde{H}_A \a{F_A^\ast}G_A\nn \\
    &= - R_M^{1/2} (I + K_M \c{F_E}G_E)\tilde{H}_A \a{F_A^\ast}\times \nn \\ & \ \ \ \ G_A\nn\\
    % &= R_M^{1/2} \tilde{H}_A \a{F_A^\ast}G_A + R_M^{1/2} K_M \c{F_E}G_E\tilde{H}_A \a{F_A^\ast}G_A\nn\\
&= - R_M^{1/2} \tilde{H}_A \a{F_A^\ast}G_A - R_M^{1/2} K_M \times\nn\\
& \ \ \ \ [\c{F_E}F_EU_3 + U_3F_A^\ast \a{F_A^\ast} + U_3]G_A,
% &= R_M^{1/2} \tilde{H}_A \a{F_A^\ast}G_A + R_M^{1/2} K_M [U_3F_A^\ast \a{F_A^\ast}]G_A+ R_M^{1/2} K_M [\c{F_E}F_EU_3 + U_3]G_A
\end{align}
where $U_3$ is the solution to the Lyapunov equation \begin{equation}
    U_3 = F_EU_3F_A^\ast + G_E\tilde{H}_A.
\end{equation}

To summarize, the strictly anticausal part of $$M^{1/2}(z)Q_2(z)W^{-/2}(z)$$ is
\begin{align}
    A(z)&= - R_M^{1/2} [ \tilde{H}_A +  K_M U_3F_A^\ast] \a{F_A^\ast}G_A \nn \\
    &\triangleq  H_A\a{F_A^\ast}G_A
\end{align}
with 
\begin{align}
    &H_A = - R_M^{1/2} (\tilde{H}_A +  K_M U_3F_A^\ast)\nn\\
    &F_A^\ast = \begin{pmatrix}
        F_W^\ast & 0\\
        U_1G_1G_1^\ast & F_S^\ast
    \end{pmatrix}, \quad G_A = \begin{pmatrix}
    H^\ast R_W^{-\ast/2}\\
    0 
    \end{pmatrix},
    \end{align}
  and \begin{equation}
      \tilde{H}_A = G_2^\ast \begin{pmatrix}
        U_1U_2&F_S^\ast 
    \end{pmatrix},
  \end{equation}  
while the causal part can presented as a sum of two causal functions
 \begin{align}
     C_1(z) &= - M^{1/2}(z) G_2^\ast U_1F_S [\c{F_S}F_SU_2 + U_2] \times \nn \\&\ \ \ \ H^\ast R_W^{-\ast/2}\\ 
     C_2(z) &= - R_M^{1/2} K_M [\c{F_E}F_EU_3 + U_3]G_A.
 \end{align}

\end{proof}

\begin{proof}[Proof of Lemma \ref{lemma:nehariSol}] This lemma is adapted from Theorem 5 in \cite{sabag2023regretoptimal} and consequently, its proof mirrors that of the aforementioned theorem. For comprehensive details, please refer to \cite{sabag2023regretoptimal}.
\end{proof}
\begin{proof}[Proof of Theorem \ref{th:SS_causal}]
Recall from Theorem \ref{th:reduction} that the Youla parameterization of the causal regret-optimal controller is:
\begin{align}\label{eq:proof_ss_c_prod3}
    Q_{c}(z)&= M^{-1/2}(z)(C_1(z) + C_2(z) + C_N(z))W^{1/2}(z)
\end{align}
where  
\begin{align}
    M^{-1/2}(z)&= (I + K_M \c{F_E}G_E)^{-1}R_M^{-/2} \nn \\
    &= (I - K_M \c{F_M}G_E)R_M^{-/2}\nn \\
    C_N(z)&= H_A \Pi(F_N \c{F_N} +I) K_N \nn \\
    W^{1/2}(z)    &= R_W^{-/2} (I - H \c{F_W}K_W) \nn \\
%  &= R_W^{-/2}(I + H \c{F}K_W)^{-1} \\
    C_1(z) &= - z M^{1/2}(z) G_2^\ast U_1F_S \c{F_S}U_2H^\ast \times \nn\\ & \ \ \ \ R_W^{-\ast/2} \\
    % &= - M^{1/2}(z) G_2^\ast U_1F_S [\c{F_S}F_SU_2 + U_2]H^\ast R_W^{-\ast/2}\nn\\
     C_2(z) &= - z R_M^{1/2} K_M \c{F_E}U_3G_A.
\end{align}
First, we write each of the (partial) products in \eqref{eq:proof_ss_c_prod3} explicitly 
\begin{align}\label{eq:proof_th_causal_each3}
    M^{-1/2}(z) C_1(z) R_W^{-/2}&= - z G_2^\ast U_1F_S \c{F_S}U_2H^\ast \times\nn\\ &\ \ \ \ R_W^{-1}\nn\\
    M^{-1/2}(z) C_2(z) R_W^{-/2} 
    % &= - z(I + K_M \c{F_E}G_E)^{-1}R_M^{-/2} R_M^{1/2} K_M \c{F_E}U_3G_AR_W^{-/2} \nn \\
    % &= - z K_M (I +  \c{F_E}G_EK_M)^{-1} \c{F_E}U_3G_AR_W^{-/2}\nn \\
    &= - z K_M \c{F_M}U_3G_AR_W^{-/2}\nn\\
    M^{-1/2}(z) C_N(z)R_W^{-/2} &= z (I - K_M \c{F_M}G_E)\times\nn\\
    &\ \ \ R_M^{-/2}H_A \Pi \c{F_N} K_N R_W^{-/2}.
\end{align}
The sum of the terms in \eqref{eq:proof_th_causal_each3} can be concisely expressed as
% We can combine the last two terms as:
% \begin{align}
%     M^{-1/2}(z) (C_2(z) + C_N(z)) R_W^{-/2} &= z \begin{pmatrix}
%     R_{M}^{-/2}H_A\Pi& - K_M 
%     \end{pmatrix} \left(zI - \begin{pmatrix}
%     F_N&0\\
%     G_ER_M^{-/2}H_A \Pi& F_M
%     \end{pmatrix} \right)^{-1} 
%     \begin{pmatrix}
%     K_NR_W^{-/2}\\
%     U_3 G_AR_W^{-/2}
%     \end{pmatrix},
%     &= z \begin{pmatrix}
%     R_{M}^{-/2}H_A\Pi& - K_M 
%     \end{pmatrix} \begin{pmatrix}
%     zI - F_N&0\\
%     - G_E R_M^{-/2}H_A \Pi& zI - F_M
%     \end{pmatrix} ^{-1} 
%     \begin{pmatrix}
%     K_NR_W^{-/2}\\
%     U_3 G_AR_W^{-/2}
%     \end{pmatrix}\\
%     &= z \begin{pmatrix}
%     R_{M}^{-/2}H_A\Pi& - K_M 
%     \end{pmatrix} \begin{pmatrix}
%     (zI - F_N)^{-1}&0\\
%     (zI - F_M)^{-1} G_ER_M^{-/2}H_A \Pi(zI - F_N)^{-1} & (zI - F_M)^{-1}
%     \end{pmatrix} 
%     \begin{pmatrix}
%     K_NR_W^{-/2}\\
%     U_3 G_AR_W^{-/2}
%     \end{pmatrix}
% \end{align}
% and the first term can be added as 
\begin{align}\label{eq:proof_th_sc_sum3}
    & z \begin{pmatrix}
    R_{M}^{-/2}H_A\Pi& - K_M & -G_2^\ast U_1 F_S
    \end{pmatrix}\times \nn \\
    &\left(zI - \begin{pmatrix}
    F_N&0&0\\
    G_ER_M^{-/2}H_A \Pi& F_M&0\\
    0&0&F_S
    \end{pmatrix} \right)^{-1} \times \nn \\ 
    &\begin{pmatrix}
    K_NR_W^{-/2}\\
    U_3 G_AR_W^{-/2}\\
    U_2 H^\ast R_W^{-1}
    \end{pmatrix}.
\end{align}
Finally, we can use the general formula:
\begin{align}\label{eq:proof_gen_formula1}
    & A \c{F}B(I - C\c{D}E) \nn\\
    &\ = 
    \begin{pmatrix}
    0& A
    \end{pmatrix}
    \c{\begin{pmatrix}
    D& 0\\
    -BC & F
    \end{pmatrix}}
\begin{pmatrix}
E\\B
\end{pmatrix},
\end{align}

% \begin{align}
%     & A \c{F}B(I - C\c{D}E) = 
%     \begin{pmatrix}
%     0& A
%     \end{pmatrix}
%     \c{\begin{pmatrix}
%     D& 0\\
%     -BC & F
%     \end{pmatrix}}
% \begin{pmatrix}
% E\\B
% \end{pmatrix},\nn
% % &= \begin{pmatrix}
% %     0& A
% %     \end{pmatrix}
% %     \begin{pmatrix}
% %     zI - D& 0\\
% %     BC & zI - F
% %     \end{pmatrix}^{-1}
% % \begin{pmatrix}
% % E\\B
% % \end{pmatrix}\\
% % &= \begin{pmatrix}
% %     0& A
% %     \end{pmatrix}
% %     \begin{pmatrix}
% %     \c{D}& 0\\
% %     - \c{F}BC\c{D} & \c{F}
% %     \end{pmatrix}
% % \begin{pmatrix}
% % E\\B
% % \end{pmatrix}\\
% \end{align}
to multiply \eqref{eq:proof_th_sc_sum3} with $(I - H \c{F_W}K_W)$ to obtain \eqref{eq:proof_ss_c_prod3}:
\begin{align}\label{eq:proof_th_causal_final}
    &Q_c(z)= z\begin{pmatrix}
    0& R_{M}^{-/2}H_A\Pi& - K_M & -G_2^\ast U_1 F_S
    \end{pmatrix}\times\nn\\
    & \left(zI - \begin{pmatrix}
    F_W & 0 & 0 &0 \\
    -K_NR_W^{-/2}H & F_N&0&0 \\
    - U_3 G_AR_W^{-/2}H & \tilde F& F_M&0 \\
    -U_2 H^\ast R_W^{-1}H & 0& 0& F_S
    \end{pmatrix}\right)^{-1}\times\nn\\
    &\begin{pmatrix}
K_W\\
K_NR_W^{-/2}\\
U_3 G_AR_W^{-/2}\\
U_2 H^\ast R_W^{-1}
\end{pmatrix},
\end{align}
with \begin{equation}
    \tilde F= G_ER_M^{-/2}H_A \Pi
\end{equation}
The proof is completed by applying the time-unit shift \begin{equation}
    z A\c{B}C = A (I + \c{B}B)C
\end{equation}to \eqref{eq:proof_th_causal_final} in order to obtain the regret-optimal controller in Theorem \ref{th:SS_causal}.
\end{proof}

\end{document}